\documentclass[12pt,a4paper]{article}
\usepackage[latin1]{inputenc}
\usepackage{amsmath}
\usepackage{amsthm}
\usepackage{amsfonts}
\usepackage{amssymb}
\usepackage{hyperref}
\usepackage{cleveref}
\usepackage{graphicx}
\usepackage{nicefrac}
\usepackage{fullpage}
\usepackage{bm}
\usepackage{xcolor}
\usepackage[linesnumbered,ruled]{algorithm2e}
\usepackage{tikz}
\usepackage{mathtools}
\usepackage[all]{xy}
\usepackage{authblk}

\newtheorem{thm}{Theorem}[section]
\newtheorem{prop}[thm]{Proposition}
\newtheorem{lemma}[thm]{Lemma}
\newtheorem{remark}[thm]{Remark}
\newtheorem{cor}[thm]{Corollary}

\newtheorem{defi}[thm]{Definition}

\newtheorem{observation}[thm]{Observation}
\newtheorem*{theoremx}{Theorem}

\newcommand*{\rom}[1]{\expandafter\@slowromancap\romannumeral #1@}

\newcommand{\eff}{\mathbb{F}}
\newcommand{\Fq}{\eff_q}

\newcommand{\Fp}{\eff_p}
\newcommand{\Fqm}[1]{\eff_{q^{#1}}}
\newcommand{\bbN}{\mathbb{N}}
\newcommand{\bbNp}{\bbN^*}
\newcommand{\bbZ}{\mathbb{Z}}
\newcommand{\cC}{\mathcal{C}}

\newcommand{\rs}{\mathsf {RS}}
\newcommand{\trs}{\widetilde{\rs}}
\newcommand{\bs}{\boldsymbol}
\newcommand{\ga}{\mathrm{GA}}
\newcommand{\clong}{C_{\mathrm{PPC}}}
\newcommand{\uclong}{\clong^{\mathrm{unfolded}}}
\newcommand{\ol}{o^{\ell}}
\newcommand{\qf}{\hat{Q}_f}
\newcommand{\veps}{\varepsilon}

\makeatletter

\newcommand{\Rmnum}[1]{\expandafter\@slowromancap\romannumeral #1@}
\makeatother

\DeclareMathOperator{\poly}{poly}

\DeclareMathOperator{\ord}{ord}
\DeclareMathOperator{\ev}{ev}
\DeclareMathOperator{\Gal}{Gal}
\DeclareMathOperator{\lcm}{lcm}

\title{Explicit Subcodes of Reed--Solomon Codes that Efficiently
Achieve List Decoding Capacity} 

\author[1]{Amit Berman}
\author[1]{Yaron Shany}
\author[2,1]{Itzhak Tamo}
\affil[1]{Samsung  Semiconductor Israel R\&D Center, 146 Derech
Menachem Begin St., Tel Aviv 6492103, Israel. Emails: \{amit.berman,
yaron.shany\}@samsung.com} 
\affil[2]{Department of Electrical Engineering-Systems, Tel Aviv
University, Tel Aviv 6997801, Israel. Email: zactamo@gmail.com}

\begin{document}

\maketitle

%%%%%%%%%%%%
% ABSTRACT %
%%%%%%%%%%%%
\begin{abstract}
In this paper, we introduce a novel explicit family of subcodes of Reed-Solomon (RS) codes that efficiently achieve list decoding capacity with a constant output list size. Our approach builds upon the idea of large linear subcodes of RS codes evaluated on a subfield, similar to the method employed by Guruswami and Xing (STOC 2013). However, our approach diverges by leveraging the idea of {\it permuted product codes}, thereby simplifying the construction by avoiding the need of {\it subspace designs}.

Specifically, the codes are constructed by initially forming the tensor product of two RS codes with carefully selected evaluation sets, followed by specific cyclic shifts to the codeword rows. This process results in each codeword column being treated as an individual coordinate, reminiscent of prior capacity-achieving codes, such as folded RS codes and univariate multiplicity codes. This construction is easily shown to be a subcode of an interleaved RS code, equivalently, an RS code evaluated on a subfield.

Alternatively, the codes can be constructed by the evaluation of bivariate polynomials over orbits generated by \emph{two} affine transformations with coprime orders, extending the earlier use of a single affine transformation in folded RS codes and the recent affine folded RS codes introduced by Bhandari {\it et al.} (IEEE T-IT, Feb.~2024). While our codes require large, yet constant characteristic, the two affine transformations facilitate achieving code length equal to the field size, without the restriction of the field being prime, contrasting with univariate multiplicity codes.

\end{abstract}

%%%%%%%%%%%%%%%%%%%%%%
\section{Introduction}
%%%%%%%%%%%%%%%%%%%%%%
Error-correcting codes are used for reliably transmitting data over
noisy communication channels. To achieve this goal, a code $C\subseteq
\Sigma^n$ (for some {\it alphabet} $\Sigma$ and {\it length} $n$) is
typically a proper subset of $\Sigma^n$, consisting only of
$|\Sigma|^k$ {\it codewords} for some $k<n$. We then say that the code
has {\it rate} $R:=k/n$, and one of the fundamental goals of coding
theory is to maximize the minimum (Hamming) distance between codewords
for a given rate $R$. 

For a code of minimum distance $d$ and {\it
normalized distance} $\delta:=d/n$, the transmitted codeword is
completely determined from the noisy channel output if the latter has
errors in less than  a fraction of  $\delta/2$ 
coordinates, while for a larger fraction of errors, the transmitted
codeword is in general not uniquely determined. Since the {\it
Singleton bound} implies that $\delta\leq 1-R$, the
maximum possible guaranteed {\it unique decoding radius} of a code of
rate $R$ is therefore $(1-R)/2$, and this is achieved by Reed--Solomon
(RS) codes. 

As originally suggested by Elias \cite{elias} and Wozencraft
\cite{wozencraft1958listdecoding}, to go beyond the 
unique decoding radius, the decoder must be allowed 
to output a  \emph{list} of potential transmitted codewords. We
say that a code $C$ is {\it $(\rho,L)$-list decodable} if for any
received word $w$, there are at most $L$ codewords $c\in C$ such that
$c$ disagrees with $w$ in at most a fraction $\rho$ of the
coordinates. 

In a breakthrough work, Sudan \cite{sudan1997decoding} presented the
first polynomial-complexity algorithm for list-decoding RS codes beyond
$(1-R)/2$ for low $R$. This was later considerably improved by
Guruswami and Sudan \cite{gurus}, who presented an efficient 
algorithm for decoding RS codes up to the {\it Johnson radius}
$1-\sqrt{R}$ with a polynomial list size. However, it is well-known
that the maximum list-decoding radius with a guaranteed constant
list size is much higher than the Johnson radius; a random coding
argument (see, e.g., \cite{guruswami2019essential}) shows that a
$(\rho,L)$-list decodable code with constant $L$ may be achieved for any 
$\rho<1-R$ (for a large enough alphabet), while it is clear that
for $\rho> 1-R$, the list 
size $L$ must be exponential. Hence, $\rho=1-R$ is called the {\it list
decoding capacity}.

Guruswami and Rudra \cite{GR08} presented the construction of {\it folded RS codes} (FRS), the first
explicit family of codes that achieve list decoding capacity. FRS codes are closely related to
RS codes; they are obtained by ``folding'' cyclic RS codewords into
the shape of matrices, where each column is considered as a single
coordinate. In fact, not only that FRS codes achieve list
decoding capacity, they achieve capacity {\it efficiently}, i.e.,  there is a
deterministic algorithm that can decode up to a fraction of
$1-R-\veps$ errors for any $\veps>0$, with 
list size and complexity polynomial in the code length. In a later work, Guruswami and Wang
\cite{guruswami-wang} presented a simple linear-algebraic decoding
algorithm for FRS codes, and proved that the output list is
contained in a subspace of dimension $O(1/\veps)$.  Hence, despite the
considerable simplification of the decoding algorithm, the list size
remained polynomial in the code length. 
To achieve a constant output list size, Guruswami and Wang proposed in \cite{guruswami-wang}  the use of pre-encoding with a combinatorial stucture they called {\it subspace evasive sets}. They also constructed these using a  probabilistic argument. Subsequently, Dvir and Lovett
\cite{dvir2012subspace}, gave an explicit construction of subspace
evasive sets. 

Recently, Kopparty {\it et
al.} \cite{kopparty-improved-list-size} revealed that the list size of
FRS codes \emph{themselves}, as well as that of other capacity
achieving codes, such as {\it univariate multiplicity codes}
(originally introduced by Rosenbloom and Tsfasman
\cite{rosenbloom1997codes}; see also
\cite{guruswami-wang} and \cite{kopparty-improved-list-size}) is in
fact constant, even without  using a pre-encoding step. Furthermore,
it was shown in \cite{kopparty-improved-list-size} 
that the constant list size can be achieved efficiently with a
randomized algorithm. Lastly, the list-size bounds of
\cite{kopparty-improved-list-size} were further tightened in
\cite{tamo2023tighter}.  

We comment that while folded RS codes univariate and multiplicity
codes have a polynomial alphabet size, there are several papers,
culminating in \cite{GRZ22}, that use constructions based on
algebraic-geometry (AG)  codes to obtain capacity-achieving codes with
both  constant list size \emph{and} constant alphabet size; see
\cite{GRZ22} and the references therein. Since this is outside the  
main scope of the current paper, we will not further elaborate on this
subject; the interested reader is referred to the introduction of
\cite{GRZ22} for a comprehensive account. 

Although folded RS codes are closely related to RS codes, they are in
general not RS codes themselves, neither (informally) large 
subcodes of RS codes. Moreover, although it is known that 
for an appropriate choice of the {\it evaluation set},
RS codes achieve list decoding capacity \emph{combinatorially}
\cite{BGM23}, \cite{GZ23}, \cite{alrabiah2023randomly} (with
exponential, quadratic, and linear finite-field size, resp.), the
evaluation sets in these works are not explicit, and there is no known
efficient algorithm for decoding these RS codes up to capacity. It is
therefore natural to ask: 

\begin{description}
\item[Q1] What is an explicit evaluation set, if exists, for
which an RS code can be efficiently list decoded up to 
capacity?  
\end{description}
As Q1 appears to be a hard question, it is also of interest to
consider the following simpler, yet non-trivial, question:
\begin{description}
\item[Q2] What is an explicit evaluation set for
which an (informally) large linear subcode of an RS code can be
efficiently decoded up to the list-decoding capacity?  
\end{description}

An answer given in \cite{GX13} to Q2 is as follows. Start with
an evaluation set that is a subfield of the finite field
over which the RS code is defined. While the resulting RS code itself
has an exponential list size (albeit with a smaller exponent than the
trivial one), it turns out that the list of coefficients of potential
information polynomials is a so-called ``periodic
subspace''. This fact is then used in \cite{GX13} to show that
pre-encoding each column of the information matrix with a different
subspace from a {\it subspace design} (also defined in
\cite{GX13}\footnote{Besides defining subspace designs, Guruswami and
Xing \cite{GX13} also gave a randomized construction of subspace design. An explicit
construction was later found by Guruswami and Kopparty
\cite{GK16}.}) 
results in an efficient capacity achieving code. Also, since the above
subspaces have a small (informally) co-dimension, this indeed results
in a large linear subcode of an RS code.

%-------------------------------
\subsection{Results and methods}
%-------------------------------
Our main contribution is a new and simple answer to Q2. In detail,
similarly to \cite{GX13}, we construct a large subcode of an
RS code evaluated on a subfield, and that efficiently achieves
list decoding capacity. However, our construction does not require the
rather involved concepts of periodic subspaces and subspace
designs. Instead, it is based on the well-known and simple construction of
(tensor) product codes. 

Informally, we start with the product of two RS codes over
$\Fq$,\footnote{$\Fq$ is the finite 
field of $q$ elements, whose characteristic is large enough. See Section \ref{sec:inst}.} 
and  coprime lengths $m$ and $n$, where the row code of length $n$ is cyclic. 
%a column  code of
%length $m$, for some $m$ to be  defined below, and a row code
%of length $n$, where the row code is cyclic, $n$ can be %arbitrarily
%larger than $m$, and $\gcd(m,n)=1$ (see Section \ref{sec:const}
%and Section \ref{sec:inst} for more details).
We then apply 
cyclic shifts to the rows of the resulting  codewords, where row $i$ is
cyclically shifted to the left by $a\cdot i$ coordinates, with
$a\equiv m^{-1}\mod(n)$. We refer to the resulting code as a {\it
permuted product code}.  

The permuted product code is obviously a subcode of the interleaved RS
code defined by requiring only that all rows are in the cyclic row
code. Moreover, it is 
well-known that the latter interleaved code is
an RS code over $\Fqm{m}$ with an evaluation set in $\Fq$, where each
entry of a codeword is replaced by the column vector of coefficients
in its  representation according to some basis of $\Fqm{m}/\Fq$. Hence showing that, indeed, the permuted product code is a subcode of an RS code. Our
main result is the following theorem.

%%%%%%%%%%%
% THEOREM %
%%%%%%%%%%%
\begin{theoremx}[Informal, see Theorem \ref{thm:mainx} below]
For $R\in(0,1)$, small enough $\veps>0$, and all powers $q$ of a prime
$p=O(1/\veps^3)$, there are instances of the permuted product code
over $\Fq$ with alphabet size $q^{O(1/\veps^3)}$,  rate $R$
and block length $q-1$, that are efficiently list-decodable from error
fraction $1-R-\veps$, with an output list of size
$(1/\veps)^{O(1/\veps^2)}$.  
\end{theoremx}
Alternatively, the permuted product code can be constructed by
evaluating bivariate 
polynomials on orbits of two elements under the action of \emph{two}
affine transformations of coprime orders. This extends the previous
usage of a single affine transformation in folded RS codes, additive
folded RS codes, and, more generally, in the recent affine folded RS
codes of Bhandari {\it et al.} \cite{BHKS23}. While folded RS
codes cannot reach a length that is as large as the size of the
underlying finite field $\Fq$ (as the code length is shorter by a
factor of the folding length), the current construction can reach a
length of $q-1$. We note that this is also 
possible with univariate multiplicity codes. However, the usage
of two affine transformations enables 
reaching a length that is as large as the finite-field size for all
rates (albeit with high, yet constant, characteristic), with a finite
field that needs not be prime.\footnote{Note that univariate
multiplicity codes have an exponential list size if the dimension is
sufficiently larger than the characteristic.}

Table \ref{table:comp} below compares the parameters of the permuted
product codes of this work with those of some other polynomial-based
capacity achieving codes. We note that while the parameters of the
construction of \cite{GX13} are better, the main contribution of the
current construction is not in its parameters, but rather that it
provides a \emph{simple} answer to Q2, constructing capacity achieving
subcodes of RS codes, using basic coding-theoretic concepts.

\begin{table}[h]
\centering
{\small
\begin{tabular}{|c|c|c|c|}
\hline
Code & Field size, $q$ & \begin{tabular}{c} %log alphabet\\ size, 
$\log_q(|\Sigma|)$\end{tabular} & \begin{tabular}{c}List size bound
%(using
%\\ \cite{tamo2023tighter})
\end{tabular} \\ 
\hline\hline
Folded RS codes & $q=O(n/ \veps^2)$ &
$O(1/\veps^2)$& $(1/\veps)^{4/\veps}$ \\
\hline
\begin{tabular}{c} Univariate\\ multiplicity codes \\
of dimension $d+1$ \end{tabular}& \begin{tabular}{c}$q=p^r$ ($r\in\bbNp$)\\ $p\geq 
\Omega(1/\veps)$ prime\\ $q\geq n$ \end{tabular} &
$O(1/\veps^2)$ & $(1/\veps)^{\frac{4}{\veps}\Big(1+\frac{d}{p}\Big)}$
\\
\hline
\begin{tabular}{c}\cite{GX13}\\ Interleaved RS\\ codes
+ \\ subspace~designs \end{tabular} & \begin{tabular}{c}$q\geq \Omega(1/\veps^2)$ \\
                 $q\geq n$ \end{tabular} & $O(1/\veps^2)$ &
$(1/\veps)^{O(1/\veps^2)}$ \\
\hline
\begin{tabular}{c} Permuted\\ product codes\\ (this work)
\end{tabular} &
\begin{tabular}{c}$q=p^r$ ($r\in\bbNp$)\\ $p\geq 
\Omega(1/\veps^3)$ prime\\ $q>n$ \end{tabular} &
$O(1/\veps^3)$  & 
$(1/\veps)^{O(1/\veps^2)}$ \\  
\hline
\end{tabular}
}
\caption{Some capacity achieving list-decodable codes of rate $R$, length $n$, and  alphabet $\Sigma$ that are list decodable from $1-R-\veps$ fraction of errors. The list size bounds are obtained via \cite{tamo2023tighter}. Specifically, the list size bound for FRS
codes is from \cite[Corollary 3.6]{tamo2023tighter}, the bound for
univariate multiplicity codes is from \cite[Theorem
3.8]{tamo2023tighter}, while the other two bounds
on the list size are obtained by using \cite[Lemma
3.1]{tamo2023tighter} with \cite[Theorem 23]{GK16} and
Theorem \ref{thm:mainx} below. 
Note that by the list size bound for univariate multiplicity codes,  the rate is positive and the list size bound is constant only if the code length is $O(p).$  
%note the dimension divided by the characteristic in the exponent
%of the list size; to remain with a constant list at all rates, $q$ must be prime. 
}  
\label{table:comp}
\end{table}

To summarize, our main contributions are as follows:

\begin{itemize}

\item We construct a new and simple large subcode of an RS code that
efficiently achieves list-decoding capacity.

\item We show that up to some cyclic shifts, the product of RS codes
can be used to achieve efficient capacity-achieving codes (where each
column of a codeword is regarded as a coordinate). This provides a new
method for constructing capacity achieving codes. 

\item We show how \emph{two} affine transformations can be used to
define capacity-achieving codes whose length is as large as the
underlying finite field.

\end{itemize}

%%%%%%%%%%%%%%%%%%%%%%%%%
\subsection{Organization}
%%%%%%%%%%%%%%%%%%%%%%%%%
In Section \ref{sec:prelim} we provide some required definitions 
 and notations. In Section \ref{sec:const} we define the permuted product codes and prove that they are
indeed subcodes of an RS code with an appropriate evaluation
set. Then, a linear-algebraic list decoding algorithm up to capacity
is presented in Section \ref{sec:listdecode}. Finally, Section
\ref{sec:conc} includes some open questions for further research. The
paper is supplemented by an appendix, in which we study the properties
of the ``unfolded'' code as a cyclic code.

%%%%%%%%%%%%%%%%%%%%%%%
\section{Preliminaries}\label{sec:prelim}
%%%%%%%%%%%%%%%%%%%%%%%
This section includes some definitions and notation that will be used
throughout the paper.

We write $\Fq$ for the finite field of $q$ elements,
where $q$ is a prime power. Throughout, we fix a prime $p$ and 
$q$ a power of $p$.  
Unless otherwise noted, all vectors are row vectors. Also,
$(\cdot)^T$ stands for matrix transposition.

%--------------------------------------------------------
\subsection{Reed--Solomon codes and their tensor products}
%--------------------------------------------------------
For integers $0\leq k\leq n\leq q$ and a set
$A=\{a_0,\ldots,a_{n-1}\}\subseteq \Fq$, the {\it Reed--Solomon (RS)
code} $\rs_{\Fq}(k,A)$ with evaluation set $A$ and dimension $k$, is
defined as  
$$
\rs_{\Fq}(k,A):=\big\{\big(f(a_0),\ldots,f(a_{n-1})\big)\big|f\in\Fq[x],
\deg(f)<k \big\}\subseteq \Fq^n.
$$
For simplicity, when the underlying finite field is clear from the
context, we will write simply $\rs(k,A)$ for $\rs_{\Fq}(k,A)$.

Next, it will be useful to recall the definition of the tensor product
of two RS codes. The {\it tensor product} $C_2\otimes C_1$ of
linear codes $C_1,C_2\subseteq \Fq^n$ of dimensions $k_1,k_2$ (resp.)
is the space of matrices whose columns are in $C_2$ and whose rows are
in $C_1$. It can be verified that $C_2\otimes C_1$ is indeed a tensor
product: it is generated as an $\Fq$-space by the outer products
$\bs{c}_2 \otimes \bs{c}_1:=\bs{c}_2^T\bs{c}_1$ for $\bs{c}_i\in C_i$,
$i=1,2$. In particular, it can be verified that if 
$\{\bs{b}'_1,\ldots,\bs{b}'_{k_2}\}$,
$\{\bs{b}_1,\ldots,\bs{b}_{k_1}\}$ are bases for $C_2$ and $C_1$
(resp.), then $\{\bs{b}'_i\otimes\bs{b}_j\}_{i,j}$ is a basis for
$C_2\otimes C_1$. It follows that $\dim(C_2\otimes C_1)=k_1k_2$, and
it is easily verified that if the minimum distances of $C_1$ and $C_2$
are $d_1,d_2$ (resp.), then the minimum distance of $C_2\otimes C_1$
is $d_1d_2$.
It also follows that for two sets  $A=\{a_i\},B=\{b_j\} \subseteq
\Fq$, and non-negative integers $s\leq |A|$, $t\leq |B|$,
$$
\rs(s,A)\otimes \rs(t,B)=\Big\{\{f(a_i,b_j) \}_{i\in\{0,\ldots,|A|-1\}
\atop j\in\{0,\ldots, |B|-1\}}\in
\Fq^{|A|\times |B|}\Big|f\in \Fq[x,y], \deg_x(f)<s,\deg_y(f)<t\Big\}. 
$$

%--------------------------------------
%\subsection{Folded Reed--Solomon codes}
%--------------------------------------

%----------------------------
\subsection{The affine group}
%----------------------------
The {\it affine group} $\ga(q)$ is the group whose underlying set is
$\{ax+b|(a,b)\in\Fq^*\times \Fq\}\subset \Fq[x]$, while the group
operation is polynomial composition: for $\ell_i:=a_ix+b_i\in \ga(q)$
($i=1,2$), $\ell_2\circ\ell_1:= \ell_2(\ell_1(x))= 
a_2a_1x+a_2b_1+b_2$. It can be verified that this is indeed a group,
with identity element $x$, and inverse
$(ax+b)^{-1}=a^{-1}x-a^{-1}b$. For $\ell(x)\in \ga(q)$ and
$i\in \bbNp$, we let
$\ell^i:=\underbrace{\ell\circ\cdots\circ\ell}_{i}$, and
$\ell^0:=x$. The {\it order} 
of $\ell$, $\ord(\ell)$, is defined as usual as the smallest $i\in
\bbNp$ such that $\ell^i=x$. 

If $\ell(x)=ax+b$ with $a\neq 1$, then
$\ell^i(x)=a^ix+b\frac{a^i-1}{a-1}$, from which it
is clear that $\ord(\ell)=\ord(a)$, where $\ord(a)$ is the order of $a$ in
$\Fq^*$. In addition, if $a=1$ and $b\neq0$, $\ell^i(x)=x+ib$, so that
$\ord(\ell)=p$. To conclude,
\begin{equation}\label{eq:ord}
\ord(ax+b)=\begin{cases}
\ord(a) & a\neq 1\\
p & a=1,b\neq 0\\
1 & a=1,b=0.
\end{cases}
\end{equation}

We let $\ga(q)$ act on $\Fq$ in the obvious way, by setting
$\ell\cdot \zeta:=\ell(\zeta)$ for $\ell\in\ga(q)$ and $\zeta\in
\Fq$. It is easily verified
that if $\zeta$ is not a fixed point of $\ell$ (i.e. $\ell(\zeta)\neq
\zeta$), then the stabilizer of $\zeta$ in the cyclic subgroup
$\langle \ell \rangle$
generated by $\ell$ is trivial, so that the orbit
$\{\ell(\zeta),\cdots,\ell^{\ord(\zeta)}(\zeta)\}$ has $\ord(\ell)$
distinct elements. This fact will be used frequently without further
mention throughout the paper. 

%---------------------------------------------------------
\subsection{The splitting field of $x^q-ax-b$}
%---------------------------------------------------------
The following properties of the splitting field of $x^q-\ell(x)$ for
$\ell\in \ga(q)$ will be useful ahead. 

%%%%%%%%%%%%%%%
% PROPOSITION %
%%%%%%%%%%%%%%%
\begin{prop}\label{prop:field}
Let $\ell(x)=ax+b\in\ga(q)$. Let $L$ be the
splitting field of $h(x):=x^q-\ell(x)$. Then
$[L:\Fq]=\ord(\ell)$. Moreover, either $(a,b)=(1,0)$ and
$L=\Fq$, or $L\supsetneq \Fq$, and $L=\Fq(\zeta)$ for any root
$\zeta$ of $h$ outside $\Fq$.
\end{prop}

\begin{proof}
First, $\ord(\ell)=1$ if and only if $a=1$ and $b=0$, in
which case $L=\Fq$. Suppose, therefore, that $(a,b)\neq (1,0)$. Since $h$ is
separable,\footnote{As $h'(x)=-a\neq 0$ is coprime to $h$.} monic,
has degree $q$, and is not equal to $x^q-x$, its splitting field is
not $\Fq$. Let $\zeta\in L\smallsetminus{\Fq}$ be a root of
$h$. 

Let $\ol:=\ord(\ell)$. We claim that
$p_{\zeta}(x):=\prod_{i=0}^{\ol-1}\big(x-\ell^i(\zeta)\big)$
is the minimal polynomial of $\zeta$ over $\Fq$. Clearly
$p_{\zeta}(\zeta)=0$. In addition, since $\zeta\notin \Fq$,
$\ell(\zeta)\neq \zeta$ (for otherwise $\zeta^q=\zeta$, as $\zeta$ is a
root of $h$), and therefore the roots 
$\{\ell^i(\zeta)=\zeta^{q^i}\}_{i=0}^{\ol-1}$ of $p_{\zeta}$ are
distinct elements\footnote{Recall that since $\zeta$ is not a fixed 
point, its orbit under the action of $\langle \ell \rangle$ has $\ol$ 
elements.} in the orbit of $\zeta$ under the action of
$\Gal(L/\Fq)$. 
Actually, the roots of $p_{\zeta}$ are an \emph{entire} orbit, as
$\zeta^{q^{\ol}}=\ell^{\ol}(\zeta)=\zeta$. This proves our claim.

Hence, for any root $\zeta$ outside $\Fq$, $[\Fq(\zeta):\Fq]=\ol$,
so all these roots lie in the same field $L=\eff_{q^{\ol}}$, as
required. 
\end{proof}

%%%%%%%%%%%%%%%%%%%%%%%%%%%    
\section{Code construction}\label{sec:const}
%%%%%%%%%%%%%%%%%%%%%%%%%%%
In this section, we first define the permuted product code as an
evaluation code. It then follows almost immediately that the code is
indeed a permuted product code, and that it is a subcode of an RS
code. 

Let $\ell_1(x),\ell_2(x)\in \ga(q)$ be two affine polynomials 
of coprime orders $m,n$, respectively, and let $\alpha,\beta\in \Fq$ be
such that $\ell_1(\alpha)\neq \alpha, \ell_2(\beta)\neq \beta$, i.e.,
$\alpha,\beta$ are not fixed points of $\ell_1$ and $\ell_2$,
respectively. 

For $f\in\Fq[x,y]$ and $j\in\{0,\ldots,n-1\}$, let
$$
\ev_j(f):=  \begin{bmatrix}
f(\ell_1^{jm}(\alpha),\ell_2^{jm}(\beta)) \\
f(\ell_1^{jm+1}(\alpha),\ell_2^{jm+1}(\beta))\\
\vdots\\
f(\ell_1^{jm+m-1}(\alpha),\ell_2^{jm+m-1}(\beta))
\end{bmatrix}\in\Fq^m,
$$
and let
$\ev\colon \Fq[x,y]\to (\Fq^m)^n$ be the function that
maps $f$ to the vector whose $j$-th entry is the column vector
$\ev_j(f)$, $j\in\{0,\ldots,n-1\}$. Explicitly, $\ev(f)$ equals
\begin{equation}  \label{encoding}
\begin{pmatrix}
 \begin{bmatrix}
f(\ell_1^0(\alpha),\ell_2^0(\beta)) \\
f(\ell_1(\alpha),\ell_2(\beta))\\
\vdots\\
f(\ell_1^{m-1}(\alpha),\ell_2^{m-1}(\beta))
\end{bmatrix}, 
 & 
  \begin{bmatrix}
f(\ell_1^m(\alpha),\ell_2^m(\beta)) \\
f(\ell_1^{m+1}(\alpha),\ell_2^{m+1}(\beta))\\
\vdots\\
f(\ell_1^{2m-1}(\alpha),\ell_2^{2m-1}(\beta))
\end{bmatrix}, 
 & \ldots,
 &
 \begin{bmatrix}
f(\ell_1^{(n-1)m}(\alpha),\ell_2^{(n-1)m}(\beta)) \\
f(\ell_1^{(n-1)m+1}(\alpha),\ell_2^{(n-1)m+1}(\beta))\\
\vdots\\
f(\ell_1^{nm-1}(\alpha),\ell_2^{nm-1}(\beta))
\end{bmatrix}
\end{pmatrix}.
\end{equation}

Let $s\leq m$ and $t\leq n$ be positive integers. 
Writing $\Fq^{s,t}[x,y]:=\{f\in \Fq[x,y]|\deg_x(f)<s,\deg_y(f)<t\}$,
the permuted product code $\clong(s,t)\subseteq (\Fq^m)^n$ is
defined as 
$$
\clong(s,t):=\big\{\ev(f)|f\in \Fq^{s,t}[x,y] \big\}.
$$

To simplify notation, we will sometimes identify $(\Fq^m)^n$ with
$\Fq^{m\times n}$, so that codewords of $\clong(s,t)$ will be regarded
either as vectors of column vectors, or as $m\times n$ matrices in the
obvious way.

Next, we would like to show that $\clong(s,t)$ is indeed a permuted
product code. Toward this end,
let $\frak{l}_2:=\ell_2^{m}$, and note that $\ord(\frak{l}_2)$ is 
$n$, since $m,n$ are coprime. Also, $\beta$ is not a fixed point of
$\frak{l}_2$.\footnote{Since $m,n$ are coprime, $\ell_2^{m}(x)\neq x$
and therefore a fixed point of $\ell_2^m$ is a fixed point of
$\ell_2$.  
%contradicting the fact that the stabilizer of $\beta$ in $\langle
%\ell_2\rangle$ is trivial.
} Let
$A:=\{\ell_1^i(\alpha)\}_{i\in\{0,\ldots,m-1\}}$,
$B:=\{\frak{l}_2^j(\beta)\}_{j\in\{0,\ldots,n-1\}}$. Finally, let
$[m]^{-1}$ be a  representative for the inverse of $m$ in
$\bbZ/n\bbZ$. Then we have the following proposition.

%%%%%%%%%%%%%%%
% PROPOSITION %
%%%%%%%%%%%%%%%
\begin{prop}\label{prop:prod}
It holds that
$$
\clong=\pi\big(\rs(s,A)\otimes \rs(t,B)\big),
$$
where $\pi$ is the permutation that shifts row $i\in\{0,\ldots,m-1\}$
to the left by $i\cdot [m]^{-1}$. 
\end{prop}

\begin{proof}
Let $\nu\in\{0,\ldots, mn-1\}$ be a ``folded running
index'' in an $m\times n$ matrix, where for a coordinate index
$(i,j)$ ($i\in\{0,\ldots, m-1\}$, $j\in\{0,\ldots, n-1\}$), we let
$\nu:=mj+i$. The 
$(i,j)$-th entry of the codeword corresponding to $f\in \Fq^{s,t}[x,y]$ is 
\begin{equation}\label{eq:genentry}
f\big(\ell_1^{\nu}(\alpha),\ell_2^{\nu}(\beta)\big) =
f\big(\ell_1^{i}(\alpha),\ell_2^{m(j+[m]^{-1}i)}(\beta)\big) =
f\big(\ell_1^{i}(\alpha),\frak{l}_2^{j+[m]^{-1}i}(\beta)\big).
\end{equation}
On the other hand, the $(i,j)$-th entry of the codeword of
$\rs(s,A)\otimes \rs(t,B)$ corresponding to $f\in \Fq^{s,t}[x,y]$ is
$f\big(\ell_1^{i}(\alpha),\frak{l}_2^{j}(\beta)\big)$, and it is clear
that \eqref{eq:genentry} corresponds to the stated cyclic shifts of the
rows.
\end{proof}
The following corollary gives the basic parameters of the permuted product code. 
%%%%%%%%%%%%%
% COROLLARY %
%%%%%%%%%%%%%
\begin{cor}\label{cor:params}
The code $\clong(s,t)$ is an $\Fq$-linear code of length $n$, rate 
$R=\frac{st}{mn}$ and minimum distance at least $n-t+1$.
\end{cor}

\begin{proof}
All assertions follow immediately from Proposition
\ref{prop:prod}. For example, for a nonzero $c\in \clong(s,t)$,
each non-zero row has weight at least $n-t+1$ as it is a nonzero
codeword of $\rs(t,B)$, therefore the number of non-zero columns is
certainly at least $n-t+1$.  
\end{proof}

Note that the code is close to being MDS if $s$ is close to $m$.

%%%%%%%%%%
% REMARK %
%%%%%%%%%%
\begin{remark}
{\rm
Some remarks are in place:
\begin{enumerate}

\item 
The construction of folded RS codes \cite{GR08} involves a
single affine polynomial, $\gamma x$, for a primitive $\gamma\in
\Fq^*$. This results in a code whose length is 
smaller than $q-1$ by a factor of the {\it folding parameter}. A
similar assertion is also true for {\it additive folded RS codes}
\cite{GR08}, \cite{BHKS23}, and for the more general {\it 
affine folded RS codes} \cite{BHKS23}, which
again use a single affine polynomial. The idea of 
using \emph{two} affine polynomials of coprime orders is a
generalization that enables to construct a capacity-achieving code of
length $q-1$, as will be shown below. 

\item The product structure can be interpreted as follows. If $s=m$,
then the vertical code is just $\Fq^m$, and the product is an
interleaved RS code \cite{SSB09}. While the interleaved code
itself does not guarantee a small list \cite{GX13}, moving
from $s=m$ to $s$ slightly smaller than $m$ (informally) results
in a guaranteed small list, as will be shown below. While in
\cite{GX13}, the non-trivial concept of  {\it subspace designs} 
was required for assuring a small list, here, 
the simpler construction of tensor product with shifts is
used.\footnote{We also note that the construction of \cite{GX13} has
some resemblance to a product code: instead of using a free matrix of
information symbols, each column is constrained to be in a different
subspace from a subspace design.} 

\end{enumerate}
}
\end{remark}
At this point, it is fairly clear that  $\clong(s,t)$ can be viewed as  a linear
subcode of $\rs_{\Fqm{m}}(t,B)$, which is an interleaved code, as
$B\subseteq \Fq$. We record this property in the following proposition. 
%%%%%%%%%%%%%%%
% PROPOSITION %
%%%%%%%%%%%%%%%
\begin{prop}
$\clong(s,t)$ can be viewed as  a linear
subcode of $\rs_{\Fqm{m}}(t,B)$.\end{prop}

\begin{proof}
Let $\trs_{\Fqm{m}}(t,B)\subset
(\Fq^m)^n$ be the code obtained by replacing each entry of each
codeword of $\rs_{\Fqm{m}}(t,B)$ by the column vector of its
coefficients in the decomposition according to a fixed  basis for $\Fqm{m}/\Fq$. Then, it is sufficient to show that   
$\clong(s,t)\subseteq \trs_{\Fqm{m}}(t,B)$.
Since the evaluation set $B$ is a subset of $\Fq$, $\trs_{\Fqm{m}}(t,B)$ is the
interleaved code whose codewords are obtained by choosing freely $m$ rows from
$\rs_{\Fq}(t,B)$, regardless of the basis choice. Note that by the
definition of $B$, $\rs_{\Fq}(t,B)$ is  cyclic, as cyclically shifting the
evaluation vector of a polynomial $f(x)$ on $B$ results
in the evaluation vector of $f(\frak{l}_2(x))$ on $B$.  Since
$\clong=\pi\big(\rs_{\Fq}(s,A)\otimes \rs_{\Fq}(t,B)\big)$ by Proposition
\ref{prop:prod}, each row of $\clong$ is a cyclic shift of a codeword
of $\rs_{\Fq}(t,B)$, and therefore again a codeword of $\rs_{\Fq}(t,B)$, and the result follows.
\end{proof}

%%%%%%%%%%%%%%%%%%%%%%%%%%%%%%%%%%%%%%%%%%%%%%%%%
\section{List decoding the permuted product code}\label{sec:listdecode}
%%%%%%%%%%%%%%%%%%%%%%%%%%%%%%%%%%%%%%%%%%%%%%%%%
Let the received, possibly corrupted, version of the codeword be 
\begin{equation}
    \label{received-word}
r=\begin{pmatrix}
    r_{0,0} & r_{0,1} & \ldots & r_{0,n-1}\\
    r_{1,0} & r_{1,1} & \ldots & r_{1,n-1}\\
    \vdots & \vdots & \ddots & \vdots\\
    r_{m-1,0} & r_{m-1,1} & \ldots & r_{m-1,n-1}
\end{pmatrix}\in \Fq^{m\times n}.
\end{equation}

The goal is to recover all polynomials $f(x,y)\in \Fq^{s,t}[x,y]$
whose encoding \eqref{encoding} agrees with $r$ on at least $\gamma$ of the
columns, for some agreement parameter $\gamma$. 
For large enough
$\gamma$, say, at least half of the minimum distance bound given in
Corollary \ref{cor:params}, the polynomial $f$, if exists, is unique. We
would like to decode beyond the unique decoding regime, i.e., for a
much smaller agreement parameter $\gamma$, by sacrificing the
uniqueness and instead outputting a list of possible codewords. To
this end, we adapt the known algebraic technique to list-decode folded
RS codes and their variants.  

%------------------------------------
\subsection{Polynomial interpolation}
%------------------------------------
In what follows, we assume that $s<m$. For a positive integer $w\leq
m-s$, consider polynomials 
of the form  
\begin{equation}
    \label{polynomial-form}
\sum_{i=0}^{w-1} p_i(x,y)z_i\text{ ,
where } \deg_x(p_i)\leq m-s-w  \text{ and }  \deg_y(p_i)\leq D-t,
\text{ for all } i,
\end{equation}
in $\Fq[x,y, z_0,\ldots, z_{w-1}]$, for some integer $D$ to be
determined later.  

The goal in the interpolation step is to interpolate a nonzero
polynomial $Q$ of the form \eqref{polynomial-form} such that for each
$0\leq j\leq n-1$, 
\begin{equation}
    \label{eq-2}
    Q(\ell_1^{jm+i}(\alpha),\ell_2^{jm+i}(\beta),r_{i,j},\ldots,
r_{i+w-1,j})=0, \text{ for } 0\leq i \leq m-w. 
\end{equation}
Note that for each $j$, the constraints \eqref{eq-2} are a  collection
of $m-w+1$ homogeneous linear constraints on the coefficients of the
polynomial $Q$, and in total there are $n(m-w+1)$ such constraints.   
The following lemma shows that a nonzero interpolation polynomial $Q$
exists and can be found efficiently.  

%%%%%%%%%
% LEMMA %
%%%%%%%%%
\begin{lemma}
With hindsight, for 
\begin{equation}
\label{value-of-D}
D:=\Big\lfloor\frac{nm}{w(m-s-w+1)} \Big\rfloor +t,
\end{equation}
 a nonzero polynomial $Q$  of the form \eqref{polynomial-form} which 
 satisfies the interpolation constraints \eqref{eq-2} exists and can
be found in $O((nm)^3)$ field operations over $\Fq$. Furthermore, we
can assume that 
$Q$ and the polynomial $y^q-\ell_2(y)$ are coprime.
\end{lemma} 

\begin{proof}
The total number of free variables in $Q$ is 
\begin{multline*}
w(D-t+1)(m-s-w+1) = w\Big(\Big\lfloor\frac{nm}{w(m-s-w+1)}
\Big\rfloor+1\Big)(m-s-w+1)\\
 > w\frac{nm}{w(m-s-w+1)}\cdot(m-s-w+1)=mn>n(m-w+1),
\end{multline*}
where the right-hand side is the number of homogeneous linear
equations for all interpolation constraints. This proves that a
non-zero $Q$ satisfying all constraints does exist, and the system of
equations (which has at most $nm$ constraints) has a
nontrivial solution that can be found efficiently.  
 
Lastly, we can assume that $Q$ and $y^q-\ell_2(y)$ are coprime, since
otherwise let $g(y)=\gcd(Q,y^q-\ell_2(y))$ and  write
$y^q-\ell_2(y)=g(y)h(y)$. We claim that $g(y)$ has no roots in the
orbit of $\beta$ under the action of $\langle\ell_2\rangle$, and
therefore the polynomial $Q/g$ satisfies too the constraints
\eqref{eq-2}.  Indeed, recall that $\beta$ was chosen to be a
non-fixed point of $\ell_2(y)$, hence also any other element $\beta'$
in the orbit  of $\beta$ under the action of $\langle \ell_2\rangle$
is too a non-fixed point. Therefore,
$$
0\neq \beta' - \ell_2(\beta')=(\beta')^q -
\ell_2(\beta')=g(\beta')h(\beta'),
$$
and the result follows. 
\end{proof}

Note that given a polynomial $Q$ of the form \eqref{polynomial-form} that
satisfies the interpolation constraints \eqref{eq-2}, it is straightforward
to modify $Q$ to be coprime to $y^q-\ell_2(y)$ while still satisfying
the constraints. This can be achieved by dividing $Q$ by any power of an 
irreducible factor of $y^q-\ell_2(y)$ that divides it. Importantly, there is no need for 
general factorization algorithms in this process, as  we focus in the
sequel on the case where  $\ell_2(x) = \gamma x$ for a primitive
$\gamma$. In such a scenario,  
$y^q - \ell_2(y) = y(y^{q-1} - \gamma)$ is the decomposition into 
irreducible factors \cite[Lemma 3.5]{GR08}.\footnote{In
the somewhat more general case where $\ell_2(x)=\gamma x+b$ with
non-zero $b$, it follows from the proof of Proposition
\ref{prop:field} that $x^q-\ell_2(x)$ factors as 
$(x-\delta)h(x)$ with $h(x)$ irreducible of degree $q-1$, and
$\delta$, the only root of $x^q-\ell_2(x)$ in $\Fq$, can be easily
found by linear algebra methods. We omit the details.}

%\zac{I think that the explanation is a bit too long. What is the
%difference between making sure $Q$ is not divisible by $y$ and $Q$ is
%not divisible by $y^{q-1}-\gamma$? they are both irreducible; you
%simply divide as much as possible. In other words, if we write
%$Q=Q'(y^{q-1}-\gamma)^a$, where is the largest such integer, then we
%simply take the new $Q$ to be $Q'$.} 
Now, dividing out the largest powers of the
irreducible factors that divide it in polynomial time is straightforward: For example,
for dividing out the largest power of $v(y):=y^{q-1} - \gamma$, it is
possible to iteratively divide all the $p_i(x,y)$ by $v(y)$, until the
first time at least one of the $p_i$'s is not divisible by it anymore.

To continue, we will need the following definition.
%%%%%%%%%%%%%%
% DEFINITION %
%%%%%%%%%%%%%%
\begin{defi}
{\rm
For a polynomial $Q$ of the form \eqref{polynomial-form}, and for
$f\in \Fq[x,y]$, we associate the bivariate polynomial
\begin{eqnarray*}
\qf(x,y)&:=&Q\big(x,y,f(x,y),f(\ell_1(x),\ell_2(y)),
\ldots,f(\ell_1^{w-1}(x),\ell_2^{w-1}(y))\big) \\
 &=& \sum_{i=0}^{w-1}p_i(x,y)f\big(\ell_1^i(x),\ell_2^i(y)\big).
\end{eqnarray*}
Consequently, if $f\in \Fq^{s,t}[x,y]$,
\begin{eqnarray}
\deg_x(\qf) &\leq& m-s-w+\deg_x(f(x,y))  < m-w, \label{eq:deg1}\\
 \deg_y(\qf)&\leq& D-t+\deg_y(f(x,y)) <  D \label{eq:deg2} 
\end{eqnarray}
}
\end{defi}

The following lemma shows the usefulness of the interpolation step for
list decoding.  

%%%%%%%%%
% LEMMA %
%%%%%%%%%
\begin{lemma}
\label{lem-interp}
Let $Q$ be a polynomial of the form \eqref{polynomial-form} that
satisfies the interpolation constraints \eqref{eq-2}. Assume that the
received word \eqref{received-word} agrees with the encoding of
$f(x,y)$ at the $j$-th coordinate for some $j\in\{0,\ldots,n-1\}$, i.e.,
$\ev_j(f)$ equals the $j$-th column of $r$. Then
$$
\qf(\ell_1^{jm+i}(\alpha),\ell_2^{jm+i}(\beta))=0 \text{ for }
i=0,\ldots,m-w.
$$
\end{lemma}

\begin{proof}
For simplicity, assume that $j=0$, and note that the general case
follows similarly.  The following is easy to verify. 
\begin{align*}
&\qf(\ell_1^i(\alpha),\ell_2^i(\beta))=\\   
& 
Q(\ell_1^i(\alpha),\ell_2^i(\beta),f(\ell_1^{i}(\alpha),\ell_2^{i}(\beta)),
\ldots,f(\ell_1^{i+w-1}(\alpha),\ell_2^{i+w-1}(\beta)))=\\
& 
Q(\ell_1^{i}(\alpha),\ell_2^{i}(\beta),r_{i,0},\ldots,r_{i+w-1, 0})=0,
\end{align*}
where the last equality follows by \eqref{eq-2}.
\end{proof}

%-------------------------------
\subsection{Outputting the list}
%-------------------------------
In this section, we present a method that uses the interpolation
polynomial in order to output the list of all polynomials
$f\in\Fq^{s,t}[x,y]$ whose encoding is close enough to the received
word $r$. 
Before we proceed, we will need the following simple lemma. 

%%%%%%%%%
% LEMMA %
%%%%%%%%%
\begin{lemma}
\label{deg-lemma}
Let $f(x,y)\in \Fq^{s,t}[x,y]$ be a polynomial. Assume that  there
exists a set $S\subseteq \Fq$ of size $s$ and a set $T_\alpha\subseteq
\Fq$ of size $t$ for any $\alpha\in S$, such that  
$$
f(\alpha,\beta)=0 \text{ for any } \alpha\in S \text{ and }
\beta\in T_\alpha.
$$ 
Then necessarily $f\equiv 0$. 
\end{lemma}

\begin{proof}
Let $f(x,y)=\sum_{i=0}^{\deg_y(f)}f_i(x)y^i$ and let $\alpha\in
S$. The univariate polynomial $f(\alpha,y)$ is of degree less than
$t$, however it vanishes on at least $t$ points, for each $\beta\in
T_\alpha$, therefore $f(\alpha,y)\equiv 0$, equivalently  $f_i(\alpha
)=0$ for any $i.$ However, $f_i(x)$ is a univariate polynomial of
degree less than $s$ that vanishes on at least $s$ points, for each
$\alpha\in S$, therefore $f_i(x)\equiv 0$ for any $i$, and the result
follows. 
\end{proof}

Assume that we have a polynomial $Q$ satisfying the interpolation
constraints.  Next, we would like to show that for a codeword that is
close enough to the received word \eqref{received-word},  the 
corresponding polynomial which generated the codeword is a root of
$Q$.  The following lemma shows exactly this. 

%%%%%%%%%
% LEMMA %
%%%%%%%%%
\begin{lemma}\label{lemma:fzeroQ} % old name: \label{lemma 2.3}
Let $f\in \Fq^{s,t}[x,y]$ be a polynomial whose encoding agrees with the
received word on at least
$$
n\Big(\frac{m}{w(m-s-w+1)}+\frac{t}{n} \Big)
$$
coordinates. Then $\qf(x,y)$ is the zero polynomial.
\end{lemma}

\begin{proof}
As before, let $\nu\in\{0,\ldots,nm-1\}$ be a running index in the codeword array,
where for row index $i\in\{0,\ldots,m-1\}$ and column index
$j\in\{0,\ldots,n-1\}$, $\nu(i,j):=mj+i$. For convenience, we will write
$i(\nu):=\nu\bmod m$ and $j(\nu):=\lfloor \nu/m\rfloor$. 

When $\nu$ runs on an entire column except for the last $w-1$
coordinates (explicitly, $\nu\in \{jm,jm+1,\ldots,jm+m-w\}$ for some
$j\in\{0,\ldots,n-1\}$), $\ell_1^{\nu}(\alpha)$ runs on the same set
$S:=\{\ell_1^i(\alpha)|i\in\{0,\ldots, m-w\}\}$ of
$m-w+1$ elements, regardless of the column $j$.
Fixing $\alpha'\in S$, the
total number of choices of $\nu$ such that: 1. $j(\nu)$ is an
agreement column, and 2. $\ell_1^{\nu}(\alpha)=\alpha'$ (and therefore
$i(\nu)\leq m-w$), is exactly the number of agreement columns, that
is, at least
\begin{equation}\label{eq:agree}
n\Big(\frac{m}{w(m-s-w+1)}+\frac{t}{n} \Big) =\frac{mn}{w(m-s-w+1)}+t\geq D.
\end{equation}
Moreover, running on these choices of $\nu$, $\ell_2^\nu(\beta)$ runs
on distinct values,\footnote{Any two distinct such choices
of $\nu$, say $\nu_2>\nu_1$, satisfy $m|(\nu_2-\nu_1)$. Since
$0<\nu_2-\nu_1<mn$ and $\gcd(m,n)=1$, we must have
$\ell_2^{\nu_1}(\beta)\neq \ell_2^{\nu_2}(\beta)$. }
and hence on a set $T_{\alpha'}$ of size at least $D$. 

Since $\qf(\alpha',\beta')=0$ for all $\alpha'\in S$ and
all $\beta'\in T_{\alpha'}$ by Lemma \ref{lem-interp},
$|S|=m-w+1>\deg_x(\qf)$, and $|T_{\alpha'}|\geq
D>\deg_y(\qf)$ for all $\alpha'$ (using
\eqref{eq:deg1}, \eqref{eq:deg2}), it follows from Lemma
\ref{deg-lemma} that $\qf=0$.  
\end{proof}

By the above lemma, we conclude that any polynomial $f$ that generates
a close-enough codeword to the received word, satisfies
$\qf=0$. Therefore, the list decoding problem boils down to efficiently
finding all such polynomials $f$ for which $\qf=0$. To this end,
we consider below a related univariate polynomial over a large
extension field $K$ of $\Fq$. 

Before proceeding, it is important to note
that as opposed to \cite{GR08}, where it is eventually
required to solve a polynomial equation over an extension field, here
$K$ is used mainly as a tool for analyzing the list size, and for
easily deriving \emph{linear-algebraic} decoding over $\Fq$-itself, as
in \cite{guruswami-wang}. See more on this in Remark
\ref{rem:linear} below. 

%%%%%%%%%%%%%%%
% PROPOSITION %
%%%%%%%%%%%%%%%
\begin{prop}\label{prop:K}
Suppose that both $m,n\neq 1$. Let $h_i(x):=x^q-\ell_i(x)$, and
$\zeta_i$ be a root of $h_i$ outside $\Fq$, $i=1,2$. Let also
$K=\Fq(\zeta_1,\zeta_2)$ be the splitting field of $h_1h_2$. Then $[K:\Fq(\zeta_1)]=n$, and
$\{\zeta_2^j\}_{j=0}^{n-1}$ is a basis for $K/\Fq(\zeta_1)$. Hence
$\{\zeta_1^i\zeta_2^j\}_{\substack{0\leq i\leq m-1\\ 0\leq j\leq  
n-1}}$ is a basis for $K/\Fq$.
\end{prop}

\begin{proof}
By Proposition \ref{prop:field}, $\Fq(\zeta_i)$ is the splitting field
of $h_i$, $i=1,2$, and we have the following diagram of field
extensions and extension degrees:
$$
\begin{gathered}
\xymatrix{
& K=\Fq(\zeta_1,\zeta_2) \ar@{-}[ld]_{d_2}
\ar@{-}[rd]^{d_1}\\
\Fq(\zeta_1)\ar@{-}[dr]_{m} && 
\Fq(\zeta_2)\ar@{-}[dl]^{n}\\ 
 & \Fq
}
\end{gathered}
$$
Since $d_2 m=d_1 n$, $d_1\leq m$, $d_2\leq n$ (as, e.g.,
the minimal polynomial of $\zeta_2$ over $\Fq(\zeta_1)$ divides
that over $\Fq$), and $\gcd(m,n)=1$, it must hold that
$d_1=m$ and $d_2=n$. Hence, $[K:\Fq(\zeta_1)]=n$, 
$\{\zeta_2^j\}_{j=0}^{n-1}$ is a basis 
for $K/\Fq(\zeta_1)$, and $\{\zeta_1^i\zeta_2^j\}_{\substack{0\leq
i\leq m-1\\ 0\leq j\leq n-1}}$ is a basis for $K/\Fq$.  
\end{proof}
%\zac{myabe we should add a remark that the linearized polynomial is
%the same as the linear algebraic approach as it appears in the book by
%sudan et al., however our approach is more abstract and less
%technical, as it does not involve dealing with matrices)} 
%Next, we define the following linearized polynomial.

%%%%%%%%%%%%%%
% DEFINITION %
%%%%%%%%%%%%%%
\begin{defi}
{\rm
Using the terminology of Proposition \ref{prop:K}, let
$$
A(z):= \sum_{i=0}^{w-1}p_i(\zeta_1,\zeta_2)z^{q^i} \in K[z].
$$
}
\end{defi}

Lemma \ref{lemma-2.4} below shows that the decoding problem reduces to
the problem of finding the roots of the \emph{linearized} polynomial
$A(z)$.\footnote{Note the substantial difference in comparison to the
situation in \cite{GR08}: there, there is a need to find the roots of an
\emph{arbitrary} polynomial over an extension field, whereas here, we
need to find the roots of a linearized polynomial, which is nothing
but solving a system of linear equations over $\Fq$ itself, similarly
to \cite{guruswami-wang}. See Remark
\ref{rem:linear} for more details.} In the lemma, we will use  the
following observation, whose omitted proof is by straightforward
induction on the $y$-degree. 

%%%%%%%%%%%%%%%
% ObSERVATION %
%%%%%%%%%%%%%%%
\begin{observation}\label{observ:longdiv}
Let $K$ be a field, let $f(x,y)\in K[x,y]$ and let $g(y)\in
K[y]$ be a non-zero polynomial. Then there exist $q(x,y),r(x,y)\in
K[x,y]$ such that: 1. $f=q(x,y)g(y)+r(x,y)$, 2. $\deg_y(r)<\deg(g)$,
3. $\deg_x(q),\deg_x(r)\leq \deg_x(f)$.  
\end{observation}

%%%%%%%%%
% LEMMA %
%%%%%%%%%
\begin{lemma}\label{lemma-2.4} The polynomial $A(z)$ satisfies the following properties. 
\begin{enumerate}

\item $A(z)$ is not the zero polynomial.

\item If $f(x,y)$ is such that $\qf(x,y)=0$, then
$A\big(f(\zeta_1,\zeta_2)\big)=0$. 

\end{enumerate}
\end{lemma}

\begin{proof}
1. Using the terminology of Proposition \ref{prop:K}, let
$p_{\zeta_2}(y)$ be the minimal 
polynomial of $\zeta_2$ over $\Fq$, which by the same proposition is
also the minimal polynomial of $\zeta_2$ over $\Fq(\zeta_1)$, and
recall that $p_{\zeta_2}(y)$ is a factor of $y^q-\ell_2(y)$. Let $i\in \{0,\ldots,w-1\}$ be such 
that 
\begin{equation}
\label{stam-equation}
     p_{\zeta_2}(y)\nmid p_{i}(x,y) \text{ (in particular, $p_{i}(x,y)\neq 0$)},
\end{equation} 
which exists by the assumption that $\gcd(Q,y^q-\ell_2(y))=1$. 

It is sufficient to prove that
$p_{i}(\zeta_1,\zeta_2)\neq 0$, equivalently,  it is sufficient to show
that $p_{i}(\zeta_1,y)$ is not divisible by $p_{\zeta_2}(y)$ in
$\Fq(\zeta_1)[y]$. Assume towards a contradiction that 
\begin{equation}
    \label{stam-equation2}
    p_{i}(\zeta_1,y)=p_{\zeta_2}(y)u(\zeta_1,y)
\end{equation} for some 
$u\in \Fq[x,y]$ with $\deg_x(u)<m$ (recall that
$[\Fq(\zeta_1):\Fq]=m$). We will show that
$p_{i}(x,y)=p_{\zeta_2}(y)u(x,y)$, a contradiction to \eqref{stam-equation}.

By Observation \ref{observ:longdiv}, write
$$
p_{i}(x,y)=p_{\zeta_2}(y)u_1(x,y)+r(x,y),
$$ 
where  
$\deg_x(u_1),\deg_x(r)\leq \deg_x(p_{i})<m$ (by the degree
assumption \eqref{polynomial-form} on $Q$), and  $\deg_y(r)<\deg(p_{\zeta_2}).$
Then, together with \eqref{stam-equation2}, 
$$p_{\zeta_2}(y)u_1(\zeta_1,y)+r(\zeta_1,y)=p_{\zeta_2}(y)u(\zeta_1,y),$$
that is,
$p_{\zeta_2}(y)(u(\zeta_1,y)-u_1(\zeta_1,y))=r(\zeta_1,y)$. This
implies that  $r(\zeta_1,y)=0$ as the $y$-degree of the right-hand
side is less than the degree of $p_{\zeta_2}$. Moreover, 
$\deg_x(r(x,y))<m=[\Fq(\zeta_1):\Fq]$, and
hence if we write $r(x,y)=\sum_i a_i(x)y^i$, $a_i(\zeta_1)\neq 0$ for
any $i$ with $a_i(x)\neq 0$. As $r(\zeta_1,y)=0$, we therefore
must have $r(x,y)\equiv 0$, and a similar argument shows also that
$u(x,y)=u_1(x,y)$, i.e., $p_{i}(x,y)=p_{\zeta_2}(y)u(x,y)$, 
and we arrive at a contradiction. 

2. Assume that $\qf(x,y)=0$. Then
\begin{multline*}
0=\qf(\zeta_1,\zeta_2) = \sum_i p_i(\zeta_1,\zeta_2)
f\big(\ell_1^i(\zeta_1),\ell_2^i(\zeta_2)\big) \\  
 = \sum_i p_i(\zeta_1,\zeta_2)
f(\zeta_1^{q^i},\zeta_2^{q^i} \big)=\sum_i p_i(\zeta_1,\zeta_2)
f(\zeta_1,\zeta_2)^{q^i}=A\big(f(\zeta_1,\zeta_2)\big).
\end{multline*}

\end{proof}
Note that for $f(x,y)\in \Fq^{s,t}[x,y]$, the polynomial $f(x,y)$ is
determined from $f(\zeta_1,\zeta_2)$, considering the basis of
Proposition \ref{prop:K}.

By combining the above results, we get the following theorem.
%\yaron{I've kept the name ``main'' in this theorem, but shouldn't the
%main theorem actually be the one in the following section, with
%concrete parameters?} 

%%%%%%%%%%%
% THEOREM %
%%%%%%%%%%%
\begin{thm}
\label{main-thm}
For every $1\leq w\leq m-s$, the permuted product code $\clong(s,t)$
satisfies that for every received word $r\in \Fq^{m\times n}$, a
subspace $W\subseteq \Fq[x,y]$ of dimension at most $w-1$ can be
found in time $\poly(\log q,m,n)$, such that every $f\in \Fq^{s,t}[x,y]$ whose
encoding \eqref{encoding} agrees with  $r$ in at least  
$n\Big(\frac{m}{w(m-s-w+1)}+\frac{t}{n} \Big)$ coordinates belongs to $W$.
\end{thm}

Note that while the theorem as stated only guarantees that the list
size does not exceed $q^{w-1}$, a general result of Kopparty {\it et
al.} \cite[Lemma 3.1]{kopparty-improved-list-size} and its recent improvement
in \cite{tamo2023tighter} can be used to move to a list size 
that does not depend on $q$. We will elaborate on this in the
following section, where we will consider a concrete choice of the
parameters for decoding up to the list decoding capacity.

%%%%%%%%%%
% REMARK %
%%%%%%%%%%
\begin{remark}\label{rem:linear}
{\rm
As an $\Fq$-vector space, $K\simeq \Fq^{mn}$. Fixing a basis (say, the basis of
Proposition \ref{prop:K}), there is some matrix $M\in\Fq^{mn\times mn}$
such that the equation $A(\zeta)=0$ takes the form $M\cdot
x^T=0$ for $x\in \Fq^{mn}$, since $A(z)$ is linearized. The
coefficients of the matrix $M$ are fixed functions of the coefficients
of the $p_i(x,y)$, similarly to the situation in
\cite{guruswami-wang}. So, although an extension field was
used as a convenient tool 
in the above derivation, this extension field does not participate in
the decoding process, and the decoder is ``linear algebraic,'' as in
\cite{guruswami-wang}. While it is perhaps possible to reach Theorem \ref{main-thm}
by constructing an appropriate triangular matrix (as in
\cite{guruswami-wang}), we have found it more convenient to
use the algebraic method described above. 
}
\end{remark}

%%%%%%%%%%%%%%%%%%%%%%%%%%%%%%%
\subsection{Code instantiation}\label{sec:inst}
%%%%%%%%%%%%%%%%%%%%%%%%%%%%%%%
In this section, we describe the selection process of the two affine
polynomials to maximize the code length for a given field
size. Subsequently, we present the criteria for parameter selection
that enables the obtained code to achieve list decoding capacity. 

We begin with the selection of the affine polynomials. By
\eqref{eq:ord}, the largest possible length $n$ is $q-1$, and then $m$
must be taken as $p$, for satisfying $\gcd(m,n)=1$ with a non-trivial
$\ell_1$. Explicitly, this 
can be achieved by setting $\ell_1(x):=x+1$ and $\ell_2(x):=\gamma x$
for a primitive $\gamma\in \Fq^*$. Now, we may take, say, $\alpha:=0$,
$\beta:=1$ for the respective non-fixed points of $\ell_1,\ell_2$.

Next, we consider parameters selection for achieving list decoding
capacity. Fix $\veps>0$. Take $m=p\approx \frac{1}{\veps^3}$,
$w\approx\frac{1}{2}\cdot\frac{1}{\veps^2}$, and $s$ such that
$m-s\approx\frac{1}{\veps^2}$ (which is indeed $\geq w$). Recall that
the normalized number of 
required agreement columns is at least $\frac{m}{w(m-s-w+1)}+\frac{t}{n}$. Let
us consider each summand 
separately. First,
$$
\frac{m}{w(m-s-w+1)} < \frac{m}{w(m-s-w)} \approx
\frac{\frac{1}{\veps^3}}
{\frac{1}{2}\cdot \frac{1}{\veps^2}\cdot \frac{1}{2}\cdot
\frac{1}{\veps^2}} = O(\veps) 
$$
Also,
$$
\frac{t}{n}=\frac{ts}{mn}\frac{m}{s} =R\frac{m}{s}\approx
R\frac{\frac{1}{\veps^3}}
{\frac{1}{\veps^3}-\frac{1}{\veps^2}}=R\frac{1}{1-\veps} \leq
R(1+2\veps)=R+O(\veps), 
$$
where the last inequality is for $\veps\leq \frac{1}{2}$. 

Summarizing the above, we may now prove our main theorem.
%%%%%%%%%%%
% THEOREM %
%%%%%%%%%%%
\begin{thm}(Main)\label{thm:mainx}
For $R\in(0,1)$, small enough $\veps>0$, and all powers $q$ of a prime
$p=O(1/\veps^3)$, there are instances of the permuted product code
over $\Fq$ with alphabet size $q^{O(1/\veps^3)}$,  rate $R$
and block length $q-1$, that are list-decodable from error
fraction $1-R-\veps$, with an output list of size
$L=(1/\veps)^{O(1/\veps^2)}$ by a randomized algorithm that outputs the list with probability at least $1-\alpha$ in time $\poly( \log \left( \frac{1}{1-\alpha}\right), 1/\varepsilon, q, L)$.
\end{thm}
Before we proceed with the proof of the theorem we will need the
following result of \cite{tamo2023tighter} specialized to our case of
list decoding. 

\begin{lemma}\cite[Lemma 3.1]{tamo2023tighter}
\label{main-lemma}
Let $\cC \subseteq (\Fq^m)^n$ be a linear code with relative  minimum
distance $\delta>0$ that is $(\delta-\varepsilon,L)$-list
decodable. Assume further that the output list size is contained in
subspace $V\subseteq \cC$ of dimension at most $r$, then the output list
size   
\begin{equation}
\label{eq:main-lemma}
    L\leq \frac{1}{\varepsilon^r}.
\end{equation}

Moreover, there is a randomized algorithm that, given a basis for $V$, with probability at least $1-\alpha$ list decodes $\cC$ with the above parameters in time $\poly(\log q, \log \left( \frac{1}{1-\alpha}\right), m, n, L)$. 
\end{lemma}

\begin{proof}[Proof of Theorem \ref{thm:mainx}]
The only part that still requires proof is the assertion regarding
the list size and the running time of the overall algorithm. By Theorem \ref{main-thm} the algorithm outputs an
$\Fq$-subspace of the permuted product code of dimension at most
$O\big( 1/\varepsilon^2\big)$ and therefore, by Lemma \ref{main-lemma}
it follows that the list is of size at most
$(1/\veps)^{O(1/\veps^2)}$. The running time follows by the running times of the deterministic algorithm in Theorem \ref{main-thm} and the randomized algorithm in Lemma \ref{main-lemma}.  
%and this follows from \cite[Lemma 3.1]{tamo2023tighter},
%given that the list is contained in an $\Fq$-subspace of dimension
%$O(1/\veps^2)$, by Theorem \ref{main-thm}.
\end{proof}

%%%%%%%%%%%%%%%%%%%%%%%%%%%%%%%%%%%%%%%%
\section{Open questions}\label{sec:conc}
%%%%%%%%%%%%%%%%%%%%%%%%%%%%%%%%%%%%%%%%
We conclude the paper with some open questions for future
research.

\begin{enumerate}

\item{{\bf More than two affine polynomials.}}
In this work, we considered codes achieving list-decoding capacity
constructed by two affine polynomials $\ell_1, \ell_2 \in \ga(q)$ of
coprime orders. This raises the question of potential benefits from
employing a larger number of affine polynomials. In particular, is it possible to construct capacity achieving codes of longer length than $q-1$, using more than two affine polynomials? Note that  if
$r\geq 2$ affine polynomials $\ell_1, \ldots, \ell_r$ are used, the
number of distinct vectors $(\ell_1^i, \ldots, \ell_r^i) \in
\ga(q)^r$, as $i$ varies over $\bbN$, is  $\lcm(\ord(\ell_1), \ldots,
\ord(\ell_r)) \leq p(q-1)$. This inequality is a consequence of
\eqref{eq:ord}, indicating that the maximum number of evaluation
points does not increase beyond that achievable with two affine
polynomials, as shown in this paper. 

However, this does not rule out the possibility of having a longer
code. For example, if $q-1=ab$ for coprime integers $a,b>1$
with $a<p$, we may take $\ell_1,\ell_2,\ell_3\in \ga(q)$ with orders
$a,p,b$ (resp.) and construct codewords with column
length $a$, similarly to \eqref{encoding}. The resulting code length will
therefore be $p b>q-1$.
%Can this idea be used to construct longer capacity achieving codes? 
%for a different arrangement of
%evaluations. In such a case, the code length might not be constrained
%to the order of a \emph{single} affine polynomial, but rather could be
%the lcm of several orders. This leads to the question: Can this
%approach be utilized to construct longer codes? 
%\yaron{I have deleted (commented out) the part asking for the
%possibility of making the characteristic independent from the column
%length, because I now believe that this is not possible: If the total
%number of evaluation points is at most $p(q-1)$ and we want to achieve
%a length larger than $q-1$, then the column length $m$ must be
%smaller than $p$. So, we must have $p>m>$some large function of
%$\veps$. What still appears to be plausible is to increase the %length
%for a similar $q$\dots I have stated this by a simple example, %which
%appears to better than stating the most general case with lcm's
%etc. However, I tend 
%to believe that we should simply omit all the ``open questions''
%section, because we haven't thought carefully and longly enough %about
%this question.}  
%More specifically, is
%it possible to design the column length $m$ as a function of $\veps$,
%the distance to capacity, independent of the characteristic $p$,
%similar to the scenario with folded RS codes? Particularly, is it
%possible to attain capacity-achieving codes of length
%$p(q-1)g(\veps)$, for some polynomial $g$, using this approach? 

\item{{\bf Using AG codes.}}
Similarly to \cite{GX13}, is it possible to extend 
%Another interesting question is regarding the potential extension of
the current results to the setup of AG codes  in order to reduce the alphabet size? 

\item{{\bf Efficient encoding.}}
Can the product structure of the construction be used for 
efficient encoding? In particular, since the horizontal code is
defined over the entire multiplicative group of $\Fq$, can this be
used for some fast evaluation algorithm? 
\end{enumerate}

\bibliographystyle{alpha}
	\bibliography{bibliography}

%%%%%%%%%%%%
% APPENDIX %
%%%%%%%%%%%%
\appendix

%%%%%%%%%%%%%%%%%%%%%%%%%%%%%%%%%%%%%%%%%%
\section{Properties of the unfolded code}
%%%%%%%%%%%%%%%%%%%%%%%%%%%%%%%%%%%%%%%%%%
In this appendix, we consider the properties of the
``unfolded'' code $\uclong(s,t)\subseteq \Fq^{mn}$, whose codewords
are the evaluation vectors
$$
\big(f(\ell_1^0(\alpha),\ell_2^0(\beta)),f(\ell_1^{1}(\alpha),\ell_2^{1}(\beta)),\ldots,
f(\ell_1^{mn-1}(\alpha),\ell_2^{mn-1}(\beta))\big),
$$ 
for all $f\in\Fq^{s,t}[x,y]$.  In particular, we show that the code is
cyclic, and we find its generator polynomial. 

It is interesting to note that while for folded RS codes,
the unfolded code is an MDS code, $\uclong(s,t)$  
is far from MDS, as its minimum distance equals $(m-s+1)(n-t+1)$
by Proposition \ref{prop:prod}. 
On the other hand, similarly to the case for folded RS codes,
$\uclong(s,t)$ is cyclic, since cyclically shifting the codeword
corresponding to $f(x,y)\in\Fq^{s,t}[x,y]$ results in the codeword
corresponding to $f(\ell_1(x),\ell_2(y))\in\Fq^{s,t}[x,y]$. It is
therefore natural to ask what is the {\it generator polynomial}
of $\uclong(s,t)$ as a cyclic code.\footnote{We will assume some
background on cyclic codes, as appearing, e.g., in 
\cite[Ch.~7]{bok:MW}, or
\cite[Ch.~8]{Roth-Ron-book}. Recall that considering cyclic codes of length $N'$
as ideals in $\Fq[x]/(x^{N'}-1)$ is valid for any $N'$, not
necessarily coprime to the characteristic. By correspondence of
ideals of $\Fq[x]$ and those of $F[x]/(x^{N'}-1)$, any 
ideal of the quotient , i.e., any cyclic code of length $N'$,
is the image of an ideal of $\Fq[x]$ containing $(x^{N'}-1)$, that
is, the image of $(g(x))$ for some $g(x)$ dividing $x^{N'}-1$. The
unique monic such $g(x)$ is called the generator polynomial of the
code. The {\it check-polynomial} is $h(x):=(x^{N'}-1)/g(x)$. It is shown,
e.g., \cite[Theorem. 7.5.4, p.~196]{bok:MW}, that the generator polynomial of the
dual code is the ``reversed $h$,'' that is, $h(0)^{-1}\cdot
x^{\deg(h)}h(x^{-1})$, and the proof remains valid when $N'$ is not coprime
to $q$.} 

For the choice of $\ell_1(x)=x+1,\ell_2(x)=\gamma x,\alpha=0,\beta=1$ from
Section \ref{sec:inst}, we answer this question in Proposition 
\ref{prop:gen} below. Note that the length $N:=mn=p(q-1)$ is
not coprime to $q$, and in general the code is a {\it repeated-root
cyclic code}, see e.g., \cite{Castagnoli}.

It will be useful to note that with the above choice of
$\ell_1,\ell_2,\alpha,\beta$, it holds that for all $f\in \Fq[x,y]$
and $\nu\in\{0,\ldots,N-1\}$, 
\begin{equation}\label{eq:eval}
f(\ell_1^{\nu}(\alpha),\ell_2^{\nu}(\beta))=f(\nu,\gamma^{\nu}).
\end{equation}

%%%%%%%%%%%%%%%
% PROPOSITION %
%%%%%%%%%%%%%%%
\begin{prop}\label{prop:gen}
The generator polynomial of $\uclong(s,t)$ is
$$
g(X):=\frac{\big(x^{q-1}-1\big)^p}
{\prod_{j={q-t}}^{q-1}(x-\gamma^j)^s}. 
$$
\end{prop}

The proof relies on the following lemma.
%%%%%%%%%
% LEMMA %
%%%%%%%%%
\begin{lemma}\label{lemma:hmat}
A polynomial $c(x)=c_0+c_1x+\cdots+c_{N-1}x^{N-1}\in \Fq[x]$ has a root
$\beta\in \Fq$ of multiplicity at least $r\leq p$ iff its vector of
coefficients $\bs{c}:=(c_0,c_1,\ldots, c_{N-1})$ satisfies
$H\bs{c}^T=\bs{0}$, where
$$
H:=\{j^i\beta^j\}_{i\in\{0,\ldots,r-1\},
j\in \{0,\ldots, N-1\}}.
$$
\end{lemma}

\begin{proof}
By the definition of the Hasse derivative (e.g., in \cite{Dvir}),
$\beta$ is a root of multiplicity at least $r$ of $c(x)$ iff the Hasse
derivatives of order $0,\ldots ,r-1$ of $c(x)$ vanish at $\beta$. As
observed in \cite{Castagnoli}, this means that having $\beta$ as
a root of multiplicity $r$ is equivalent to $H_0\cdot
\bs{c}^T=\bs{0}$, where 
$$
H_0:=\left\{\binom{j}{i}\beta^j\right\}_{i\in\{0,\ldots,r-1\},
j\in \{0,\ldots, N-1\}}.
$$
Note that the row index $i$ satisfies $i\leq r-1\leq p-1$. Now, for
$i<p$, $i!$ is invertible in $\Fp\subseteq \Fq$, and  
$$
\binom{j}{i} = (i!)^{-1}\cdot j(j-1)\cdots(j-i+1)
$$
(note that this holds also for $j<i$, where $\binom{j}{i}=0$). 
In $\Fp[x]$, write $a_0+a_1x+\cdots+a_ix^i:=(i!)^{-1}\cdot
x(x-1)\cdots(x-i+1)$, 
where $a_i=(i!)^{-1}\neq 0$,
and set $\bs{a}:=(a_0,\ldots, a_i)$.
Then $\binom{j}{i}$ is the $\bs{a}$-linear combination of
$1,j,\ldots,j^i$ (same $\bs{a}$ for all $j$), and therefore row $i$
of $H_0$ is the $\bs{a}$-linear combination of rows $0,\ldots,i$ of
$H$. Hence, $H_0$ is obtained by multiplying $H$ from the left by an
invertible lower triangular matrix, and it follows that both matrices
have the same row space.   
\end{proof}

With the lemma, the proof of Proposition \ref{prop:gen} is now
straightforward: 
\begin{proof}[Proof of Proposition \ref{prop:gen}]
A generator matrix of $\uclong(s,t)$ can be obtained by evaluating the
monomials $x^iy^j$, $i\in\{0,\ldots,s-1\},j\in\{0,\ldots,t-1\}$, as
defined in (\ref{eq:eval}), for $\nu\in\{0,\ldots,N-1\}$. The
resulting matrix in $\Fq^{st\times N}$ has rows
$$
\left\{\nu^i(\gamma^{j})^{\nu}\right\}_{\nu\in\{0,\ldots,N-1\}}
$$
for all $i\in\{0,\ldots,s-1\},j\in\{0,\ldots,t-1\}.$
Noting that $s\leq p$, it follows from Lemma \ref{lemma:hmat} that the
generator polynomial of the dual code is 
$g^{\perp}(x):=\prod_{j=0}^{t-1} (x-\gamma^j)^s$, so that the check
polynomial of $\uclong(s,t)$ 
is 
$$
h(x):=\prod_{j=0}^{t-1} (x-\gamma^{-j})^s =
\prod_{j={q-t}}^{q-1}(x-\gamma^j)^s. 
$$
\end{proof}

%\bibliographystyle{alpha}
%	\bibliography{bibliography}

%%%%%%%
% END %
%%%%%%%
\end{document}